\documentclass[letterpaper, 10 pt, conference]{ieeeconf}  

\IEEEoverridecommandlockouts                              

\usepackage{cite}
\usepackage{amsmath,amssymb,amsfonts}
\usepackage{graphicx}
\usepackage{textcomp}

\usepackage[utf8]{inputenc}
\usepackage{graphics} 
\usepackage{epsfig} 
\usepackage{times} 

\usepackage{graphicx}      
\usepackage{siunitx}
\usepackage{pgfplots}
\usepackage{tikz,amsmath}
\usetikzlibrary{external}
\usepackage{tikzscale}
\usepackage{algpseudocode}
\usepackage{amsthm}
\usepackage{algorithm}
\usepackage{fancyhdr}
\usepackage{subcaption} 
\usepackage{eso-pic}
\usepackage{mathdots}
\usepackage{accents}

\usetikzlibrary{shapes,arrows}
\usetikzlibrary{positioning}

\let\labelindent\relax
\usepackage{enumitem}

\newtheorem{thm}{Theorem}
\newtheorem{lem}[thm]{Lemma}
\newtheorem{defn}[thm]{Definition}

\theoremstyle{plain}

\pgfplotsset{compat=newest}
\pgfplotsset{plot coordinates/math parser=false}
\newlength\figureheight
\newlength\figurewidth

\newlength\defcolwidth

\newcommand*{\tran}{^{\mkern-1.5mu\mathsf{T}}}

\usepackage{tikz}
\usepackage{tikzscale}

\definecolor{myorange}{cmyk}{0,0.35,0.85,0} 
\definecolor{mypurple}{cmyk}{0.5,1,0,0} 

\definecolor{matblue1}{rgb}{0,0.4470,0.7410}
\definecolor{matred1}{rgb}{0.85,0.325,0.098}
\definecolor{matyel1}{rgb}{0.9290, 0.6940, 0.1250}
\definecolor{matpur1}{rgb}{0.4940, 0.1840, 0.5560}
\definecolor{matgre1}{rgb}{0.4660, 0.6740, 0.1880}
\definecolor{matblue2}{rgb}{0.3010, 0.7450, 0.9330}
\definecolor{matred2}{rgb}{0.6350, 0.0780, 0.1840}
\definecolor{matgrey1}{rgb}{0.5, 0.6, 0.7}
\definecolor{matpink1}{rgb}{1, 0.07, 0.65}
\definecolor{matblue3}{rgb}{0.07, 0.62, 1}

\newcommand{\yeldots}{\raisebox{2pt}{\tikz{\draw[-,matyel1,densely dotted,line width = 0.9pt](0,0) -- (3mm,0);}}}

\newcommand{\reddash}{\raisebox{2pt}{\tikz{\draw[-,matred1,dashed,line width = 0.9pt](0,0) -- (3mm,0);}}}

\newcommand{\blueline}{\raisebox{2pt}{\tikz{\draw[-,matblue1,solid,line width = 0.9pt](0,0) -- (3mm,0);}}}
\newcommand{\redline}{\raisebox{2pt}{\tikz{\draw[-,matred1,solid,line width = 0.9pt](0,0) -- (3mm,0);}}}

\newcommand{\yelline}{\raisebox{2pt}{\tikz{\draw[-,matyel1,solid,line width = 0.9pt](0,0) -- (3mm,0);}}}

\graphicspath{{./Figures/}}

\title{\LARGE \bf
Cross-Coupled Iterative Learning Control for Complex Systems: \\ A Monotonically Convergent and Computationally Efficient Approach*
}%

\author{Leontine Aarnoudse$^{1}$, Johan Kon$^{1}$, Koen Classens$^{1}$, Max van Meer$^{1}$, \\ Maurice Poot$^{1}$, Paul Tacx$^{1}$, Nard Strijbosch$^{2}$ and Tom Oomen$^{1,3}$
	\thanks{*This work is part of the research programme VIDI with project number 15698, which is (partly) financed by the NWO.}
	\thanks{$^{1}$The authors are with the Dept. of Mechanical Engineering, Control Systems Technology, Eindhoven University of Technology, Eindhoven, The Netherlands. {\tt\small l.i.m.aarnoudse@tue.nl}}%
	\thanks{$^{2}$Nard Strijbosch is with IBS Precision Engineering, Eindhoven, The Netherlands}%
	\thanks{$^{3}$Tom Oomen is also with the Delft Center for Systems and Control, Delft University of Technology, Delft, The Netherlands.
		}%
}%

\begin{document}
	\AddToShipoutPictureBG*{%
	\AtPageUpperLeft{%
		\setlength\unitlength{1in}%
		\hspace*{\dimexpr0.5\paperwidth\relax}
		\makebox(0,-1)[c]{
			\parbox{\paperwidth}{ \centering
				Leontine Aarnoudse, Johan Kon, Koen Classens, Max van Meer, Maurice Poot, Paul Tacx, Nard Strijbosch and Tom Oomen, \\ Cross-Coupled Iterative Learning Control for Complex Systems: A Monotonically Convergent and Computationally Efficient Approach, \\
				To appear in {\em Conference on Decision and Control 2022}, Cancún, Mexico, 2022}}%
}}

\maketitle%
\thispagestyle{empty}%
\pagestyle{empty}%

\setlength\defcolwidth{7.85cm}

\setlength\figurewidth{.9\defcolwidth}
\setlength\figureheight{.7\figurewidth}

\begin{abstract}
Cross-coupled iterative learning control (ILC) can achieve high performance for manufacturing applications in which tracking a contour is essential for the quality of a product. The aim of this paper is to develop a framework for norm-optimal cross-coupled ILC that enables the use of exact contour errors that are calculated offline, and iteration- and time-varying weights. Conditions for the monotonic convergence of this iteration-varying ILC algorithm are developed. In addition, a resource-efficient implementation is proposed in which the ILC update law is reframed as a linear quadratic tracking problem, reducing the computational load significantly. The approach is illustrated on a simulation example.
\end{abstract}

\section{Introduction} \label{sec:intro}

%
%
%
%

The demands for accuracy and speed in manufacturing are ever increasing and necessitate improved control strategies. In many multiple-input multiple-output (MIMO) applications such as (3D) printing and CNC machining, accurately following a contour is essential for the quality of the final product. For increasingly complex parts with curved surfaces and sharp curvature variation, high feedrates lead to high contour errors, resulting in low processing qualities and efficiency \cite{Jia2018}. Standard control approaches that aim at following a time-based reference are not suitable for this type of applications.

Approaches to minimize contour errors include offline methods such as trajectory generation and pre-compensation \cite{Yang2015} and online methods such as improved control of individual axes and cross-coupled feedback control \cite{Koren1980,Huo2012a}. Cross-coupled feedback control uses an online estimate of the contour error, which is typically a reference-dependent combination of the individual axes errors. This introduces time-varying couplings between the system axes, leading to linear time-varying (LTV) or non-linear systems, for which feedback control is not straightforward \cite{Shih2002}. In addition, stability and computation times limit the accuracy of online contour error estimates for cross-coupled feedback control.

Iterative learning control (ILC) is capable of achieving high accuracy by updating a feedforward signal iteratively based on repeated experiments \cite{Bristow2006a}. Updating the input using previous error signals leads to high performance after only a small number of experiments. Typical frameworks include frequency-domain \cite{Blanken2020a} and lifted norm-optimal ILC \cite{Gunnarsson2001,Owens2016}. Typical norm-optimal ILC aims at reducing the deviation from a time-based reference signal for individual axes by minimizing a cost function that includes individual axes error and input signals. While this leads to high performance for the individual axes, standard ILC has limitations for contour tracking applications. In particular, it often leads to extremely large inputs in sharp corners, where it may be preferable to reduce the speed instead, and to limited velocities in straight parts due to the time-based reference. For contour tracking, the performance of ILC can be increased significantly by explicitly taking the contour error into account.

Several aspects are of importance for cross-coupled ILC. First, high accuracy of the contour error estimate is a requirement for precise control. Online contour error estimates vary in accuracy and computational load, and range from linear or circular \cite{Koren1980} to more complex parameter-based approximations \cite{Chen2016}. Since feedforward signals in ILC are calculated offline, the exact contour error can be used instead of these approximations. The exact contour error is constructed from the individual axes errors through coupling gains that follow from comparing an output coordinate to each point on the reference contour. These exact coupling gains are not only time-varying but also iteration-varying, and as such require an ILC framework that allows for variation over iterations.

Second, cross-coupled control should consider not only the contour error but also a tangential error to ensure that the system moves over the contour. For norm-optimal cross-coupled ILC, it is essential that the tuning of the design parameters of the algorithm is intuitive and reflects the trade-off between accuracy and speed represented by respectively the contour and tangential errors. In addition, the relative importance of contour and tangential errors typically differs over the reference, e.g., in corners accuracy is more important than speed. The framework should enable time-varying weighting to allow for these variations. An approach to cross-coupled ILC is proposed in \cite{Barton2008a}, in which standard ILC and feedback control for the individual axes are combined with PD-based ILC based on the contour error. In \cite{Barton2011} the method is extended to norm-optimal ILC with time-varying weighting. Both approaches use linear contour error estimates, limiting the achievable performance. In addition, while the norm-optimal ILC approach weights both individual axes and contour errors, the tuning is not intuitive since the trade-off between tangential and contour errors is not made explicit. 

Third, the considered MIMO applications often use long reference signals. However, standard norm-optimal lifted ILC methods are limited to reference signals with a relatively small number of samples, because the ILC update law involves matrix inversions for which the computational load scales badly with the signal length \cite{Zundert2016}. For cross-coupled ILC, a low-order solution with efficient computations that allows for long reference signals is essential. The implementation in \cite{Barton2011} requires inversion of large matrices and the approach is therefore limited to short reference signals. To remove restrictions on the reference length, non-lifted cross-coupled ILC is developed in \cite{Sun2014a}. This approach uses linear contour error approximations in individual cost functions at each error sample, and results in a PD-like ILC controller that does not allow for a direct feedthrough term in the plant.

Although several important steps have been taken towards cross-coupled ILC for contour tracking, a framework that uses exact contour errors and enables intuitive tuning and unlimited signal lengths is lacking. This paper aims to address these aspects, resulting in the following contributions. 
\begin{itemize}
	\item A cost function is introduced that allows the use of exact contour errors with iteration-varying coupling matrices, and that enables intuitive tuning of time- and iteration-varying weights (Section \ref{sec:approach}).
	\item Conditions for monotonic convergence of the ILC algorithm with the proposed iteration-varying cost function are given (Section \ref{sec:convergence}).
	\item A resource-efficient implementation based on linear quadratic tracking is proposed that allows for inexpensive and fast computations in case of iteration-varying ILC matrices and long reference signals (Section \ref{sec:implementation}).
\end{itemize}
The approach is illustrated using a simulated flatbed printer in Section \ref{sec:sims}. Conclusions are given in Section \ref{sec:conclusions}.

\section{Problem formulation} \label{sec:problem}

In this section, the contour tracking problem is defined and the norm-optimal ILC framework is introduced.

\subsection{Contour tracking}

Consider a discrete-time, linear time-varying (LTV) MIMO system with $n_i$ inputs and $n_o$ outputs. For reference $y_d$, the tracking error $e$ is given by
\begin{align} \label{eq:error}
	e = S y_d - J f,
\end{align}
with sensitivity $S = (I+PC)^{-1}$ for plant $P$ and controller $C$, process sensitivity $J = P(I + CP)^{-1}$, and feedforward input $f$ as shown in Fig. \ref{fig:par_ilc}. The system $J$ and input $f$ are given in lifted form by
\begin{align}\nonumber 
J &= \begin{bmatrix} H^{0,0} & \dots & 0 \\ \vdots & \ddots & \vdots \\ H^{N-1,0} & \dots & H^{N-1,N-1} \end{bmatrix},  \: f = \begin{bmatrix} f(0) \\ f(1) \\ \vdots \\ f(N-1) \end{bmatrix}, 
\end{align} 
with $J$ the convolution matrix of the LTV system, which is a lower triangular matrix with a block Toeplitz structure and entries $H^{i,j} \in \mathbb{R}^{n_o \times n_i}$. The error $e\in \mathbb{R}^{N n_o \times 1}$ is written similar to input $f\in \mathbb{R}^{N n_i \times 1}$, and $e(k) \in \mathbb{R}^{n_o \times 1}$ and $f(k) \in \mathbb{R}^{n_i \times 1}$ are of the form
$	e(k) = \begin{bmatrix}
		e^1(k) & e^2(k) & \dots & e^{n_o}(k)
	\end{bmatrix}\tran. $
The aim of the system is to track a contour described by the reference $y_d(k) \in \mathbb{R}^{n_o \times 1}$ accurately in space rather than in time. To that end, in addition to the time-based error $e(k)$, the contour error $\varepsilon_c(k) \in \mathbb{R}$ is defined as the distance between the position output $y(k) \in \mathbb{R}^{n_o \times 1}$ and the closest point on the contour, as illustrated in Fig. \ref{fig:contour_error}. Each contour error sample $\varepsilon_c(k)$ can be expressed as a function of $e(k)$ through a vector of coupling gains $c(k) \in \mathbb{R}^{n_o \times 1}$ according to
\begin{align}
	\varepsilon_c(k) = c(k)\tran e(k).
\end{align}
The coupling gains depend on the type of approximation used to determine $\varepsilon_c$, as is further explained in Section \ref{sec:approach}.

\begin{figure}
	\centering
	\includegraphics{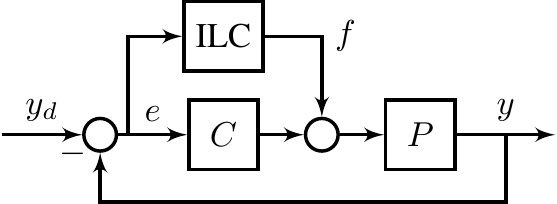}%
	\caption{Parallel ILC configuration. \label{fig:par_ilc}}%
	\vspace{-10pt}
\end{figure}

\begin{figure}
	\centering
	\includegraphics[width=.9\linewidth]{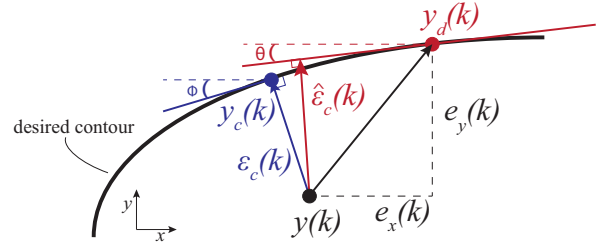}%
	\caption{For a 2D-system with reference $y_d(k)$, position output $y(k)$ and individual axes errors $e_x(k)$ and $e_y(k)$, the actual contour error $\varepsilon_c(k)$ (\protect\blueline) based on the closest point $y_c(k)$ on the contour differs from linear approximation $\hat{\varepsilon}_c(k)$ (\protect\redline). \label{fig:contour_error}}%
	\vspace{-15pt}
\end{figure}

\subsection{Norm-optimal iterative learning control}

In norm-optimal iterative learning control (ILC), the input $f$ of system (\ref{eq:error}) is updated iteratively according to
\begin{align} \label{eq:error_j}
	e_j = Sy_d - J f_j, \\
	f_{j+1} = Q f_j + L e_j, \label{eq:ff_j}
\end{align}
where error $e_j$ and $f_j$ depend on iteration $j$. The update matrices $Q$ and $L$ follow from the cost function
\begin{align} \label{eq:cost}
	\mathcal{J}(f_{j+1}) = \|e_{j+1}\|^2_{W_e} + \|f_{j+1}\|^2_{W_f} + \|f_{j+1} - f_j\|^2_{W_{\Delta f}},
\end{align}
with $\|x\|_{W} = \sqrt{x\tran W x}$. The minimizer of $\mathcal{J}(f_{j+1})$ can be determined analytically \cite{Gunnarsson2001} and leads to 
\begin{align} 
	Q &= (J\tran W_e J + W_f + W_{\Delta f})^{-1}(J\tran W_e J + W_{\Delta f})\\
	L &= (J\tran W_e J + W_f + W_{\Delta f})^{-1}J\tran W_e.
\end{align}
If the convolution matrix $J$ is non-singular, norm-optimal ILC leads to monotonic convergence of the sequence of input signals $\{f_j\}$ for $W_e \succ 0$, $W_f, W_{\Delta f} \succeq 0$. If $J$ is singular then $W_f \succ 0$ enforces monotonic convergence.

This paper aims to develop a framework for norm-optimal cross-coupled ILC that achieves high performance for contour tracking applications. A time- and iteration-varying cost function is developed that allows the use of exact contour errors as well as various estimations and that enables intuitive tuning. The convergence is analyzed and a resource-efficient implementation suitable for long references is proposed.

\section{Cost function design} \label{sec:approach}

In this section, a cost function is introduced that encompasses different configurations of cross-coupled ILC. It enables using exact contour errors with iteration-varying coupling gains and intuitive time- and iteration-varying weights.

\subsection{Cost function framework}
The following cost function is proposed, in which different approaches to cross-coupled ILC can be incorporated
\begin{align} \label{eq:cost2}
	\mathcal{J}(f_{j+1}) = \|e_{j+1}\|^2_{W_{ec,j}} + \|f_{j+1}\|^2_{W_{fc,j}} + \|f_{j+1} - f_j\|^2_{W_{\Delta fc,j}}. 
\end{align}
The weighting-coupling matrices $W_{ec,j}$, $W_{fc,j}$ and $W_{\Delta fc,j}$ may be iteration-varying and can take into account individual axes inputs and errors, different approximations of the contour error, and the tangential error that is perpendicular to the contour error. The matrices are of the form
\begin{align}
	W_{ec,j} &= C_{e,j}\tran W_{e,j} C_{e,j}, \quad
	W_{fc,j} = C_{f,j}\tran W_{f,j} C_{f,j}, \\ \nonumber
	W_{\Delta fc,j} &= C_{f,j}\tran W_{\Delta f,j} C_{f,j}
\end{align}
where the coupling matrices $C_{e,j}$ for $e$ and $C_{f,j}$ for $f$ may differ, for example when the MIMO system is non-square or different couplings for input and error are desired. The block-diagonal matrices are constructed as
\begin{align} \label{eq:coupling}
	C_{e,j} = \begin{bmatrix}
		C^1_{e,j} & \dots & 0 \\
		\vdots & \ddots & \vdots \\
		0 & \dots & C^N_{e,j}
	\end{bmatrix}.
\end{align}
The blocks $C^k_{e,j} \in \mathbb{R}^{n_{ce} \times n_o}$, with $n_{ce}$ the number of coupled error components, describe the coupling of the axes at each sample and have full column rank to ensure that each output component is taken into account. Standard norm-optimal ILC is recovered for $C^k_{e,j} = I^{n_i\times n_o} \:\forall k,j$. Possible configurations for a 2D system with $e(k) = \begin{bmatrix} e_y(k) & e_x(k)\end{bmatrix}\tran$ include:
\begin{itemize}[leftmargin=*]
	\item Cross-coupled ILC with linear approximations of the contour error $\hat{\varepsilon}_c$ and individual axes errors, with $\theta(k)$ the angle between reference sample $y_d(k)$ and the $x$-axis:
	\begin{align}
		C^k_{e,j} = \begin{bmatrix} 1 & 0 \\ 0& 1 \\ \cos(\theta(k)) & -\sin(\theta(k))  \end{bmatrix} \: \forall j,
		\end{align} 
	\item Cross-coupled ILC with linear approximations of the contour and tangential errors:
	\begin{align}
		C^k_{e,j} = \begin{bmatrix} \cos(\theta(k)) & -\sin(\theta(k)) \\ \sin(\theta(k)) & \cos(\theta(k))  \end{bmatrix} \: \forall j.
	\end{align} 
	\item Cross-coupled ILC with exact contour and tangential errors, with $\phi(k,j)$ the angle between the x-axis and the vector perpendicular to the contour error vector. The contour error vector is the vector from position $y_j(k)$ to the closest point on the contour $y_{c,j}(k)$, see Fig. \ref{fig:contour_error}. In this approach, the coupling gains may differ over iterations, resulting in iteration-varying weighting-coupling matrices:
	\begin{align} \label{eq:exact_cc}
		C^k_{e,j} = \begin{bmatrix} \cos(\phi(k,j)) & -\sin(\phi(k,j)) \\ \sin(\phi(k,j)) & \cos(\phi(k,j))  \end{bmatrix},
	\end{align} 
\end{itemize}
The matrix $C_{f,j}$ with blocks $C^k_{f,j} \in \mathbb{R}^{n_{cf} \times n_o}$ is designed similar to $C_{e,j}$. For systems with $n_o > 2$ the exact contour error definition can be extended easily, but approximations become more complicated \cite{Jia2018}. 

\subsection{Time- and iteration-varying weighting} \label{subsec:timevar}
The structure of cost function (\ref{eq:cost2}) allows for time-varying weighting matrices $W_{e,j}$, $W_{f,j}$ and $W_{\Delta f,j}$. The matrices have the following block-diagonal structure:
\begin{align}
	W_{e,j} = \begin{bmatrix}
		W^1_{e,j} & \dots & 0 \\
		\vdots & \ddots & \vdots \\
		0 & \dots & W^N_{e,j}
	\end{bmatrix},
\end{align}
where the size of $W^k_{e,j} \in \mathbb{R}^{n_{ce} \times n_{ce}}$ depends on that of the corresponding block $C^k_{e,j}$ of the coupling matrix. The diagonal blocks $W^k_{e,j}$ apply weights to each element of the coupled error. For example, for $C^k_{e,j}$ according to (\ref{eq:exact_cc}), i.e., exact contour and tangential errors, it holds that
\begin{align}
	C^k_{e,j}  \begin{bmatrix} e_y(k) \\ e_x(k)\end{bmatrix} = \begin{bmatrix} \varepsilon_c(k) \\ \varepsilon_t(k) \end{bmatrix}, 
\end{align}
such that the first diagonal element of $W^k_{,je}$ puts a weight on the contour error $\varepsilon_c(k)$ and the second diagonal element puts a weight on the tangential error $\varepsilon_t(k)$. The blocks $W^k_{e,j}$ need not be identical for all $k$, allowing time-varying weighting such as increased weights on the contour error in corners to increase accuracy. The weights may also be iteration-varying, e.g., one could weight only individual axes errors initially and add weights on the contour error after some iterations.

\section{Monotonic convergence} \label{sec:convergence}

In this section, conditions are developed for the monotonic convergence towards a closed 2-norm ball of cross-coupled ILC with the iteration-varying cost function of Section \ref{sec:approach}. The ILC update based on cost function (\ref{eq:cost2}) is given by
\begin{align} \label{eq:f_update}
	f_{j+1} =& Q_j f_j + L_j e_j, \: \text{with} \\
	Q_j =& (J\tran W_{ec,j} J + W_{ef,j} + W_{\Delta fc,j})^{-1} \\ \nonumber 
	&(J\tran W_{ec,j} J + W_{\Delta fc,j}), \\ \label{eq:filters}
	L_j =& (J\tran W_{ec,j} J + W_{ef,j} + W_{\Delta fc,j})^{-1}J\tran W_{ec,j}. 
\end{align}
The sets containing all possible filters $Q_j$ and $L_j$ are given by $\mathcal{Q}$ and $\mathcal{L}$, respectively. Since the reference and position are sampled with finite resolution, the sets $\mathcal{Q}$ and $\mathcal{L}$ are finite provided that the set of iteration-varying weights is also finite. First, monotonic convergence towards a closed $2$-norm ball is defined, then a convergence theorem is given.

\begin{defn}[Closed $2$-norm ball]
	The closed 2-norm ball $B_2(c,d)$ with center $c\in\mathbb{R}$ and radius $d \in \mathbb{R}_{\geq 0}$ is defined as $B_2(c,d) := \{x\in \mathbb{R} | \| x-c\|_2 \leq d\}$.
\end{defn}

\begin{defn}[Monotonic convergence towards a closed $2$-norm ball]
	The sequence $\{y_i\}$, $y_i \in \mathbb{R}$ is said to converge monotonically in the 2-norm to the 2-norm ball $B_2(c,d)$ if there exists $\kappa \in [0,1)$ such that for all $i \in \mathbb{Z}_{\geq 0}$,
	\begin{align}
		&\|y_{i+1} - c\|_2 \leq \kappa \|y_i -c\|_2 \quad \text{if } \: y_i \notin B_2(c,d), \\
		&y_{j+1} \in B_2(c,d) \qquad \qquad \quad \: \text{if } \: y_i \in B_2(c,d).
	\end{align}
\end{defn}

\begin{thm} \label{thm:conv}
	The sequence of inputs $\{f_j\}$ that follows from update law (\ref{eq:f_update}), with iteration-varying filters $Q_j \in \mathcal{Q}$ and $L_j \in \mathcal{L}$ that minimize criterion (\ref{eq:cost2}) according to (\ref{eq:filters}), is monotonically convergent towards a closed $2$-norm ball if the coupling matrix $C_{e,j}$ has full column rank, $W_{e,j} \succ 0 \: \forall j$ and either,
	\begin{itemize}
		\item if $J$ is non-singular, $W_{f,j}, W_{\Delta f,j} \succeq 0 \: \forall j$, or, 
		\item if $J$ is singular, $C_{f,j}$ has full column rank and $W_{f,j} \succ 0, W_{\Delta f,j} \succeq 0 \: \forall j$.
	\end{itemize}
\end{thm}
\noindent
The following auxiliary lemma is used in the proof.

\begin{lem} \label{lem:conv_ball}
	For iteration (\ref{eq:f_update}), the following two statements are equivalent:
	\begin{enumerate}
		\item The sequence of inputs $\{f_j\}$ with fixed $Q_j = \bar{Q} \in \mathcal{Q}$ and $L_j = \bar{L} \in \mathcal{L}$ for all $j$ is monotonically convergent in the 2-norm to a fixed point.
		\item The sequence of inputs $\{f_j\}$ with iteration-varying $Q_j \in \mathcal{Q}$ and $L_j \in \mathcal{L}$ is monotonically convergent in the 2-norm towards a closed 2-norm ball.
	\end{enumerate}
\end{lem}
\begin{proof}[Proof of Lemma \ref{lem:conv_ball}]
	To show that $1) \implies 2)$, assume first that for each $\bar{Q} \in \mathcal{Q}$ and $\bar{L} \in \mathcal{L}$ the sequence of inputs $\{f_j\}$ of the corresponding iteration-invariant ILC system converges monotonically to a fixed point $\bar{f}_\infty$, i.e.,
	\begin{align} \label{eq:proof_eq1}
		\|f_{j+1} - \bar{f}_\infty\|_2 \leq \kappa \|f_j - \bar{f}_\infty\| \: \forall \: j,
	\end{align}
	for some universal $\kappa \in [0,1)$. Consider $f_j \in B_2(c,d)$ and given $Q_j \in \mathcal{Q}, L_j \in \mathcal{L}$. It holds that $B_2(c,d) \subset B_2(\bar{f}_{j,\infty},d+\|c-\bar{f}_{j,\infty}\|_2)$ because
	\begin{align}
		\|f_{j} - \bar{f}_{j,\infty}\|_2 \leq \|f_{j} -c\|_2 + \|c-\bar{f}_{j,\infty}\|_2.
	\end{align}
	Therefore, if $f_j \in B(c,d)$, i.e., $\|f_{j} -c\|_2 \leq d$, then
	\begin{align}
		\|f_{j} - \bar{f}_{j,\infty}\|_2 \leq d + \|c-\bar{f}_{j,\infty}\|_2.
	\end{align}
	It follows from (\ref{eq:proof_eq1}) that $f_j \in B_2(\bar{f}_{j,\infty},d+\|c-\bar{f}_{j,\infty}\|_2) \implies f_{j+1} \in B_2(\bar{f}_{j,\infty},\kappa d+ \kappa \|c-\bar{f}_{j,\infty}\|_2)$, which leads to 
	\begin{align*}
		\|f_{j+1} - \bar{f}_{j,\infty}\|_2 \leq \kappa (d+\|c-\bar{f}_{j,\infty}\|_2).
	\end{align*}
	It holds that $B_2(\bar{f}_{j,\infty},\kappa d+ \kappa \|c-\bar{f}_{j,\infty}\|_2) \subset B_2(c, \kappa d + (1+\kappa)\|c-\bar{f}_{j,\infty}\|_2)$ because
	\begin{align}
		\|f_{j+1} - c\|_2 &\leq \|f_{j+1} - \bar{f}_{j,\infty}\|_2 + \|c-\bar{f}_{j,\infty}\|_2 \\ \nonumber
		&\leq \kappa (d+\|c-\bar{f}_{j,\infty}\|_2) + \|c-\bar{f}_{j,\infty}\|_2. 
	\end{align} 
	It follows that $f_j \in B_2(c,d) \implies f_{j+1} \in B_2(c, \kappa d + (1+\kappa)\|c-\bar{f}_{j,\infty}\|_2) \: \forall \: Q_j \in \mathcal{Q}, \: L_j \in \mathcal{L}$. Next, consider the set $B(c,d^*)$ with
	\begin{align}
		d^* = \max_{\bar{Q} \in \mathcal{Q}, \bar{L} \in \mathcal{L}} \kappa d + (1+\kappa) \|c-\bar{f}_\infty\|_2.
	\end{align}
	Thus if $f_j \in B_2(c,d)$, then $f_{j+1} \in B_2(c,d^*) \: \forall \: Q_j \in \mathcal{Q}, \: L_j \in \mathcal{L}$. Since the sets $ \mathcal{Q}$ and $ \mathcal{L}$ are finite and $\kappa \in [0,1)$, there exists $a$ such that $d > d^*$ if $d > a$, and $d=d^*$ if $d = a$.

	To show that $2) \implies 1)$, assume that the sequence of inputs $\{f_j\}$ is monotonically convergent in the 2-norm to a closed 2-norm ball given by $B_2(c,a)$, i.e.,
	\begin{align} \label{eq:conv_proof1}
		\|f_{j+1} - c\|_2 \leq \kappa \|f_j - c\|_2 \: \text{if} \: \|e_j-c\|_2 > a,
	\end{align}
	Since this is satisfied for any iteration-varying $Q_j \in \mathcal{Q}, \: L_j \in \mathcal{L}$, the iteration-invariant case where $Q_j = \bar{Q}, \: L_j = \bar{L} \: \forall j$ also satisfies (\ref{eq:conv_proof1}). Therefore, each of the iteration-invariant systems converges monotonically towards the closed 2-norm ball $B_2(c,a)$ and since they are iteration-invariant, they converge monotonically to a fixed point in this set.
\end{proof}
\noindent

\begin{proof}[Proof of Theorem \ref{thm:conv}]
	The proof consists of four steps. 
	
	\noindent
	\textit{Step 1.} The sequence of inputs $\{f_j\}$ in (\ref{eq:f_update}) for fixed $\bar{Q} \in \mathcal{Q}$ and $\bar{L} \in \mathcal{L}$ is monotonically convergent in the 2-norm if the mapping from $f_j$ to $f_{j+1}$ is a contraction mapping according to the Banach fixed point theorem \cite[Theorem 5.1-2]{Kreyszig1978}. Substituting (\ref{eq:error}) in (\ref{eq:f_update}) shows that this is satisfied if
	\begin{align} \label{eq:proof_cond}
		\|\bar{Q} - \bar{L} J\|_2 < 1.
	\end{align}	
	\textit{Step 2.} It holds that $\bar{\sigma} (\bar{Q} - \bar{L} J) \leq \|\bar{Q} - \bar{L} J\|_2$. From (\ref{eq:filters}) it follows that (\ref{eq:proof_cond}) is satisfied if
	\begin{align*}
		\bar{\sigma} (\bar{Q} - \bar{L} J) = \bar{\sigma}((J\tran \bar{W}_{ec} J + \bar{W}_{fc} + \bar{W}_{\Delta fc})^{-1} \bar{W}_{\Delta fc}) < 1.
	\end{align*}
	It holds that $\bar{\sigma}((A+B)^{-1} B) < 1$ for $A \succ 0, \: B \succeq 0$. For a positive (semi)definite matrix $M$, $A\tran M A$ is positive (semi)definite if $A$ has full column rank. Thus for non-singular $J$, $\bar{W}_{ec} \succ 0$, $\bar{W}_{fc}, \bar{W}_{\Delta fc}\succeq 0$ ensures monotonic convergence. For singular $J$, $\bar{W}_{fc} \succ 0$ is needed also.	
	
	\noindent
	\textit{Step 3.} Matrices $W_{ec}$, $W_{fc}$ and $W_{\Delta fc}$ are structured as $C\tran W C$. Thus, $\bar{W}_{ec}, \bar{W}_{fc} \succ 0$ is satisfied if $C_{e}$ respectively $C_{f}$ has full column rank and $W_{e}$ respectively $W_{f} \succ 0$. Additionally, $\bar{W}_{\Delta fc}\succeq 0$ is satisfied for $W_{\Delta f} \succeq 0$.
	
	\noindent	
	\textit{Step 4.} Applying Step 1-3 for each $Q_j \in \mathcal{Q}$, $L_j \in \mathcal{L}$ and combining with Lemma \ref{lem:conv_ball} concludes the proof.	
\end{proof}
Note that for iteration-invariant weighting-coupling matrices, Theorem \ref{thm:conv} ensures monotonic convergence in the 2-norm to a fixed point instead. Theorem \ref{thm:conv} reduces the design of cross-coupled ILC for monotonic convergence to choosing suitable weights and couplings, enabling intuitive design. It is possible to find an expression for the smallest closed 2-norm ball to which the system converges, a result used in a preliminary version of \cite{Strijbosch2022}, see \cite[Theorem III.9]{Strijbosch2019g}.

\section{Resource-efficient implementation} \label{sec:implementation}

In this section, the cross-coupled ILC update law (\ref{eq:f_update}) is reframed as a linear quadratic tracking problem with a resource-efficient solution, reducing the computational load significantly and enabling long reference signals. This is especially useful for iteration-varying cost functions, since  $Q_j$ and $L_j$ are not calculated explicitly and the inversion in (\ref{eq:filters}), that would otherwise limit the size of the lifted matrices and thus the length of the reference signal, is avoided. 

Process sensitivity $J$ is rewritten to state-space description  
\begin{align}
	\textbf{J}_{SP} = \left[\begin{array}{c|c} 
		A& B\\ 
		\hline 
		C & D 
	\end{array}\right]. 
\end{align}
Due to the block-diagonal structure of $W_{ec,j}$, $W_{fc,j}$ and $W_{\Delta fc,j}$ the cost function (\ref{eq:cost2}) can be written as
\begin{align} \label{eq:cost3}
	\mathcal{J}(f_{j+1}) = \sum_{k=1}^{N} &\|e_{j+1}(k)\|^2_{W^k_{ec,j}} +  \|f_{j+1}(k)\|^2_{W^N_{fc,j}}  \\ \nonumber &+ \|f_{j+1}(k) - f_j(k)\|^2_{W^k_{\Delta fc,j}}, 
\end{align}
where, consistent with the previous notation, $W^k_{ec,j} = (C^k_{e,j})\tran W^k_{e,j} C^k_{e,j}$ etc. From (\ref{eq:error_j}) it follows that $e_{j+1} =  e_j - J (f_{j+1}-f_j)$. In addition, $\Delta f_{j+1}(k) = f_{j+1}(k) - f_j(k)$ and $\Delta e_{j+1}(k) = e_{j+1}(k) - e_j(k)$ are defined. This leads to the following theorem that relates the optimal input $f_{j+1}$ in (\ref{eq:cost2}) to the solution of a linear quadratic tracking problem. 

\begin{thm} \label{thm:riccati}
	The optimal ILC input that minimizes (\ref{eq:cost2}) with $W_{\Delta fc,j} \succ 0$ is the solution to the linear quadratic tracking problem with cost function
	\begin{align} 
		&\mathcal{J}(\Delta f_{j+1}) = \sum_{k=1}^{N} \Delta f_{j+1}\tran(k) \underbrace{W^k_{\Delta fc,j}}_{R^k_j} \Delta f_{j+1}(k) +\\ \nonumber
		&  \begin{bmatrix}
			e_j(k) + \Delta e_{j+1}(k) \\ f_j(k) + \Delta f_{j+1}(k)  \end{bmatrix}\tran \underbrace{\begin{bmatrix} W^k_{ec,j} & 0 \\ 0 & W^k_{fc,j} \end{bmatrix}}_{S^k_j} \begin{bmatrix}
		    e_j(k) + \Delta e_{j+1}(k) \\ f_j(k) + \Delta f_{j+1}(k) \end{bmatrix}, 
	\end{align}	
	subject to the dynamics
	\begin{align}
	\Delta x_{j+1}(k+1) &= A \Delta x_{j+1}(k) + B \Delta f_{j+1}(k), \\
	\Delta y_{j+1}(k) &= \begin{bmatrix} \Delta e_{j+1}(k) \\ \Delta f_{j+1}(k) \end{bmatrix} \\ \nonumber & = \begin{bmatrix} C \\ 0 \end{bmatrix} \Delta x_{j+1}(k) + \begin{bmatrix} D \\ I \end{bmatrix} \Delta f_{j+1}(k).
	\end{align} 	
\end{thm}

\begin{proof}
	The proof follows from substituting $e_{j+1} =  e_j - J (f_{j+1}-f_j)$ and $\Delta e_{j+1} = -J(f_{j+1}-f_j) = -J \Delta f_{j+1}$ in (\ref{eq:cost3}). The problem is reframed as an LQT problem with direct feedthrough by adding output $\Delta f_{j+1}$. Taking $W_{\Delta fc,j} \succ 0$ ensures that $S^k_j \succeq 0, \: R^k_j \succ 0$ and concludes the proof. 
\end{proof}
The solution in Theorem \ref{thm:riccati} is identical to the lifted ILC update (\ref{eq:f_update}), in contrast to the non-lifted approach in \cite{Sun2014a} which minimizes an individual cost function at each sample and as such is fundamentally different. The solution to discrete-time LQT problems with $S^k_j \succeq 0, \: R^k_j \succ 0$ is well-known, see \cite[Section 4.4]{Lewis2012} or \cite{Ebrahimzadeh2017} for the situation with direct feedthrough. A Hamiltonian system is defined, leading to a two-point boundary value problem to which a sweep method is applied. For cross-coupled ILC, define
\begin{align} \nonumber 
	\bar{C} = \begin{bmatrix} C \\ 0 \end{bmatrix}, \: \bar{D} = \begin{bmatrix} D \\ I \end{bmatrix}, \:
	r_j(k) = \begin{bmatrix}-e_j(k) \\ -f_j(k) \end{bmatrix}. 
\end{align}
The optimal input is $f_{j+1}(k) = f_j(k) + \Delta f_{j+1}(k)$, where
\begin{align}
	\Delta f_{j+1}(k) = &-\bar{R}_j^{-1}(k) \bar{P}_j(k+1) \Delta x_{j+1}(k) \\ \nonumber  & + \bar{R}_j^{-1}(k) \bar{D} \tran S^k_j r_j(k) + \bar{R}_j^{-1}(k) B\tran v(k+1), 
	\end{align} \begin{align}
	\text{with} \quad \bar{R}_j(k) =& R^k_j + \bar{D}\tran S^k_j \bar{D} + B\tran P_j(k+1) B, \\
	\bar{P}_j(k+1) =& B\tran P_j(k+1) A + \bar{D}\tran (S^k_j)\tran \bar{C}.
\end{align}
The terms $v_j(k+1)$ and $P_j(k+1)$ follow from solving the following equations backwards in time:
\begin{align}
	P_j(k) =& A\tran P_j(k+1) A + \bar{C}\tran S^k_j \bar{C} - \\ \nonumber
	& (A\tran P_j(k+1) B + \bar{C}\tran S^k_j \bar{D}) \bar{R}^{-1}_j(k) \bar{P}_j(k+1), \\
	v_j(k) =& - \left(\bar{P}_j\tran(k+1) \bar{R}^{-1}_j(k)  \bar{D}\tran - \bar{C}\tran \right) S^k_j r_j(k)  \\ \nonumber &- \left(\bar{P}_j\tran(k+1) \bar{R}^{-1}_j(k) B\tran - A\tran \right) v_j(k+1),
\end{align}
with $x(0) = 0$ and boundary conditions
\begin{align}
	P_j(N) &= \bar{C}\tran S^N_j \bar{C}, \\
	v_j(N) &= \bar{C}\tran S^N_j(\bar{D} \Delta f_{j+1} - r_j(N)). 
\end{align} 
Compared to Theorem \ref{thm:conv}, Theorem \ref{thm:riccati} also requires $W_{\Delta fc,j} \succ 0$. This is not limiting in practice, as $W_{\Delta fc,j} \succ 0$ is also required to limit the amplification of iteration-varying disturbances \cite{Oomen2017}. A similar low-order solution to ILC, which omits the explicit formulation of the ILC update law as an LQT problem, is applied to the specific cases of ILC for intersample behavior in \cite{Oomen2011} and norm-optimal ILC in \cite{Zundert2016}.

\section{Example} \label{sec:sims}

\begin{figure}
	\centering
	\setlength\figurewidth{.8\defcolwidth}
	\setlength\figureheight{.6\figurewidth}
	\includegraphics{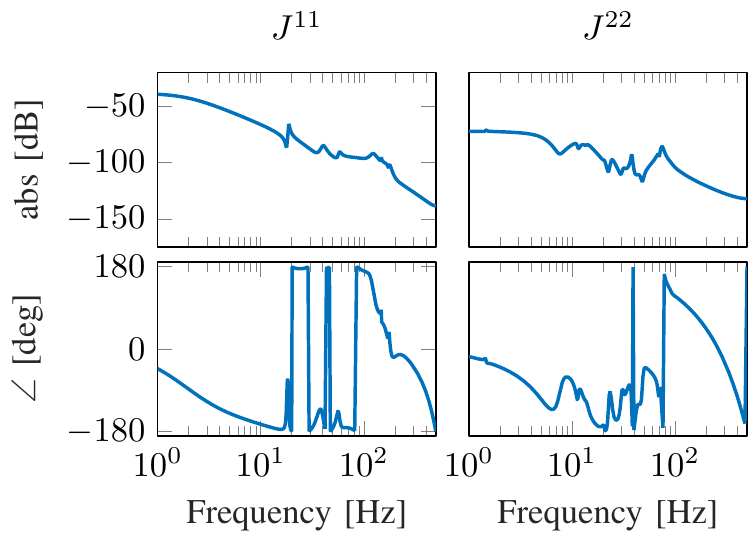}%
	\caption{Bode diagram of the process sensitivity of a decoupled industrial flatbed printer. Input and output 1 correspond to the y-axis, input and output 2 correspond to the x-axis.  \label{fig:bode}}%
	\vspace{-15pt}
\end{figure}
In this section, the proposed cross-coupled ILC framework is illustrated using simulations of an industrial flatbed printer, the process sensitivity of which is shown in Fig. \ref{fig:bode}. Cross-coupled ILC is applied to this system with coupling matrices $C_{e,j} = C_{f,j}$ and coupling gains representing the exact contour and tangential errors according to (\ref{eq:exact_cc}). The time- and iteration-invariant weights are given by
\begin{align} \nonumber
	W^k_{e,j} &= \begin{bmatrix} 1.5 & 0 \\ 0 & 0.5 \end{bmatrix}, \: W^k_{f,j} = \begin{bmatrix} 10^{-10} & 0 \\ 0 & 10^{-9} \end{bmatrix}, \\ \nonumber 
	W^k_{\Delta f,j} &= 10^{-10} I^{2\times 2},  \qquad \forall k,j.
\end{align}
Thus, the weight on the contour error $\varepsilon_c$ is higher than that on the tangential error $\varepsilon_t$, and the input penalty in the direction of the contour error is smaller, reflecting the aim of minimizing the contour error in this application.

The ILC algorithm with iteration-varying cross-coupling matrices for exact contour errors converges, see Fig. \ref{fig:contours} and \ref{fig:convergence}. In this case, the main contour error reduction comes from reducing the error in $x$-direction, see also the individual axes errors in Fig. \ref{fig:convergence}. Comparisons with other norm-optimal ILC cost functions are omitted, because different weights lead to different trade-offs between control input and error, making such comparisons completely arbitrary. The advantage of the proposed framework is in the use of exact contour errors, which ensures that the cost function represents the aim of contour tracking, and in the intuitiveness of tuning.

\begin{figure}
	\centering
	\setlength\figurewidth{.8\defcolwidth}
	\setlength\figureheight{.65\figurewidth}
	\includegraphics{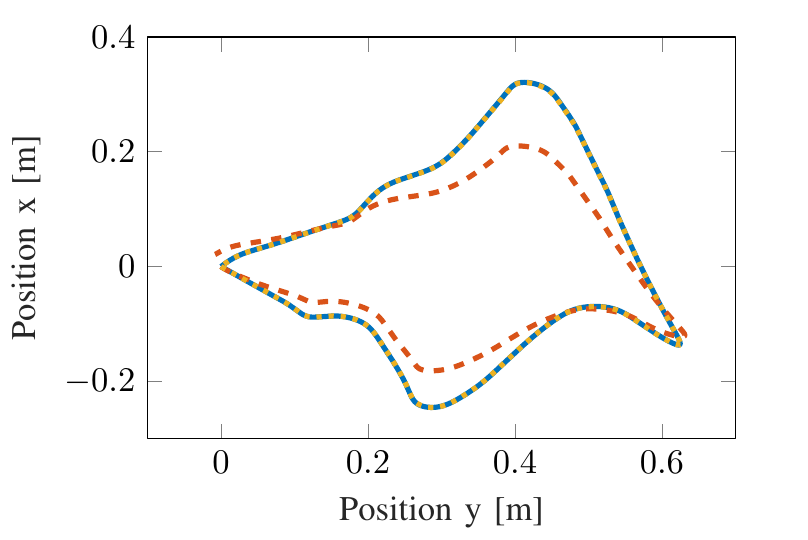}%
	\vspace{-5pt}
	\caption{Reference (\protect\blueline), position output without feedforward (\protect\reddash) and position output with a significantly reduced contour error after ten cross-coupled ILC iterations (\protect\yeldots). \label{fig:contours}}%
	\vspace{-10pt}
\end{figure}

\begin{figure}
	\setlength\figurewidth{.8\defcolwidth}
	\setlength\figureheight{.45\figurewidth}
	\centering
	\includegraphics{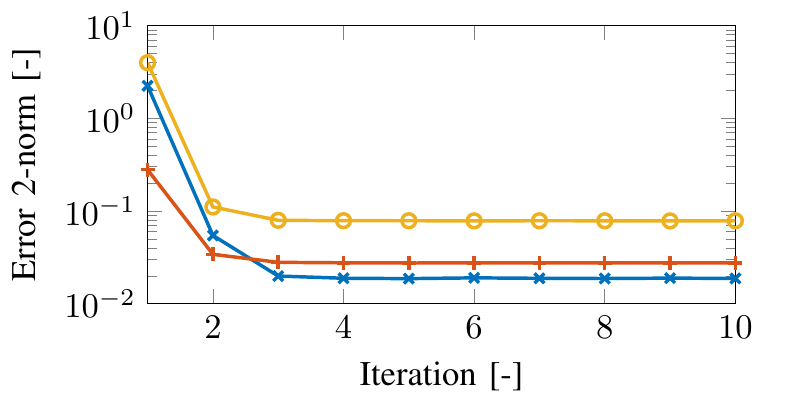}%
	\vspace{-5pt}
	\caption{Convergence of the 2-norm of the contour error (\protect\blueline) and the errors of the individual axes $y$ (\protect\redline) and $x$ (\protect\yelline). \label{fig:convergence}}%
	\vspace{-15pt}
\end{figure}

\section{Conclusion} \label{sec:conclusions}

In this paper a new framework for cross-coupled iterative learning control is developed that enables the use of exact contour errors and time- and iteration-varying weights. Conditions for the monotonic convergence of the ILC algorithm are given. In addition, the ILC update law is reframed as a linear quadratic tracking problem,  which can be solved efficiently for arbitrary long reference signals. The approach is illustrated on a simulation example of a flatbed printer; experimental results that confirm the theoretical and simulation results are omitted due to space limitations. Future work includes developing design criteria for time-varying weights.

\section*{Acknowledgment}
The authors gratefully acknowledge the contributions to this paper through a challenge-based learning project by Dirk Alferink, Gijs van den Brandt, Emre Deniz, Robert Devillers, Mike van Duijnhoven, Roel Habraken, Daan den Hartog, Shaun Boyteen Joseph, Jord van Kalmthout, Boudewijn Kempers, Sjoerd Leemrijse, Walter MacAulay, Paul Munns, Aron Prinsen, Stan de Rijk, Sander Ruijters, Jos Snijders, Chuck Steijlen, Jaap van der Stoel, Matthijs Teurlings, Hugo Thelosen, Peter Visser, Naomi de Vos and Matthijs van de Vosse. The authors also wish to thank Sjirk Koekebakker for his contributions.


\addtolength{\textheight}{-12cm}   



\bibliography{IEEEabrv,library}

\end{document}